%% file: main2.tex
\documentclass[english, 12pt]{article}

\usepackage{geometry}
\usepackage{enumitem}
\usepackage{array}
\usepackage{relsize}
\usepackage{stmaryrd}
\usepackage{fullpage}
\usepackage{fourier}
\usepackage[T1]{fontenc}
\usepackage{calrsfs}
\DeclareMathAlphabet{\pazocal}{OMS}{zplm}{m}{n}
\usepackage{microtype}
\usepackage{multirow} 
\usepackage{cite}
\usepackage{amsfonts}
\usepackage{amssymb}
\usepackage{amstext}
\usepackage{amsmath}
\usepackage{mathtools}
\usepackage{paralist}

\usepackage{tikz}
\usetikzlibrary{calc}
\usetikzlibrary{shapes,arrows,positioning,shadows,snakes}

\usepackage{footmisc}

\usepackage{amssymb,amsmath,amsthm,amsfonts}
\newtheorem{theorem}{Theorem}
\newtheorem{lemma}{Lemma}
\newtheorem{conjecture}{Conjecture}

\usepackage{pifont}
\usepackage[hang,small,bf]{caption}
\usepackage{subcaption}

\usepackage{nameref}
\usepackage[linktocpage=true,pagebackref=true]{hyperref}
\usepackage{cleveref}

\usepackage{thmtools,thm-restate} % See section 1.4 of the pdf above

\usepackage{graphicx}
\usepackage{graphics}
\usepackage{colordvi}
\usepackage{xspace}
\usepackage{algorithm}
\usepackage{algorithmicx}
\usepackage{algpseudocode}
\usepackage{url}
\usepackage{enumitem}
\usepackage{authblk}

\newcommand*{\shifttext}[2]{%
  \settowidth{\@tempdima}{#2}%
  \makebox[\@tempdima]{\hspace*{#1}#2}%
}

\newcommand\TombStone{\rule{.7ex}{1.7ex}}
\renewcommand{\qedsymbol}{\TombStone}
\newcommand{\qedd}{\let\qed\relax\quad\raisebox{-.1ex}{$\qedsymbol$}}

        {\hspace*{\fill}$\Box$\par\vspace{4mm}}

\newtheorem{proposition}[theorem]{Proposition}

\newcommand{\I}{{\pazocal I}}

\newcommand{\be}{\begin{enumerate}}
\newcommand{\ee}{\end{enumerate}}
\newcommand{\bd}{\begin{description}}
\newcommand{\ed}{\end{description}}
\newcommand{\bi}{\begin{itemize}}
\newcommand{\ei}{\end{itemize}}

\renewcommand{\phi}{\varphi}

\newcommand{\Ac}{\pazocal{A}}
\newcommand{\Rc}{\pazocal{R}}
\newcommand{\OPTR}{\mathsf{OPT^R}}
\newcommand{\MIR}{\mathsf{MIR}}
\newcommand{\OPTM}{\mathsf{OPT^M}}
\newcommand{\OPTML}{\mathsf{OPT_{\boxslash}^M}}
\newcommand{\OPTMR}{\mathsf{OPT_{\boxbslash}^M}}
\newcommand{\OPTS}{\mathsf{OPT^S}}
\newcommand{\OPTSp}{\mathsf{OPT^{S}_{\boxslash}}}
\newcommand{\OPTRp}{\mathsf{OPT^{R}_{\elbr}}}
\newcommand{\OPTRm}{\mathsf{OPT^{R}_{\elbl}}}
\newcommand{\OPTSm}{\mathsf{OPT^{S}_{\boxbslash}}}
\newcommand{\elbl}{\mbox{(\rotatebox[origin=c]{90}{\reflectbox{$\llcorner$}} \rotatebox[origin=c]{90}{\reflectbox{$\urcorner$}}\!\!)}\xspace}
\newcommand{\elbr}{\mbox{(\!\!\reflectbox{\rotatebox[origin=c]{90}{\reflectbox{$\urcorner$}}} \reflectbox{\rotatebox[origin=c]{90}{\reflectbox{$\llcorner$}}}\hspace{.07em})}\xspace}

\newcommand{\elba}{\mbox{(\rotatebox[origin=c]{90}{\reflectbox{$\llcorner$}}\hspace{.25em}\reflectbox{\rotatebox[origin=c]{90}{\reflectbox{$\llcorner$}}}\hspace{.07em})}\xspace}

\newcommand{\elbb}{\mbox{(~\rotatebox[origin=c]{90}{\rotatebox[origin=c]{90}{\reflectbox{$\llcorner$}}\hspace{.25em}\reflectbox{\rotatebox[origin=c]{90}{\reflectbox{$\llcorner$}}}\hspace{.07em}})}\xspace}

\newcommand{\elbc}{\mbox{(\rotatebox[origin=c]{180}{\rotatebox[origin=c]{90}{\reflectbox{$\llcorner$}}\hspace{.25em}\reflectbox{\rotatebox[origin=c]{90}{\reflectbox{$\llcorner$}}}\hspace{.07em}})}\xspace}

\newcommand{\elbd}{\mbox{(\rotatebox[origin=c]{270}{\rotatebox[origin=c]{90}{\reflectbox{$\llcorner$}}\hspace{.25em}\reflectbox{\rotatebox[origin=c]{90}{\reflectbox{$\llcorner$}}}\hspace{.07em}}~)}\xspace}

\newcommand{\elbul}{\mbox{(\!\!\reflectbox{\rotatebox[origin=c]{90}{\reflectbox{$\urcorner$}}}\hspace{.07em})}\xspace}

\newcommand{\elbur}{\mbox{(\,\rotatebox[origin=c]{90}{\reflectbox{$\urcorner$}}\!\!)}\xspace}

\newcommand{\elbdr}{\mbox{(\,\reflectbox{\rotatebox[origin=c]{90}{\reflectbox{$\llcorner$}}}\hspace{.07em})}\xspace}

\newcommand{\elbtr}{\mbox{(\rotatebox[origin=c]{270}{\rotatebox[origin=c]{90}{\reflectbox{$\llcorner$}}\hspace{.25em}\reflectbox{\rotatebox[origin=c]{90}{\reflectbox{$\llcorner$}}}\hspace{.07em}}\rotatebox[origin=c]{90}{\hspace{.25em}\reflectbox{\rotatebox[origin=c]{90}{\reflectbox{$\llcorner$}}}\hspace{.07em}})}\xspace}

\newcommand{\elbdl}{\mbox{(\rotatebox[origin=c]{90}{\reflectbox{$\llcorner$}}\,)}\xspace}

\newcommand{\OPTsi}{\mathsf{OPT^{sS}}}
\newcommand{\OPTri}{\mathsf{OPT^{sR}}}
\newcommand{\diam}{\mathsf{diam}}

\def\markatright#1{\leavevmode\unskip\nobreak\quad\hspace*{\fill}{#1}}
\renewenvironment{proof}
  {\begin{trivlist}\item[\hskip\labelsep{\emph{Proof}.}]}
  {\markatright{\qed}\end{trivlist}}

\def\ShowComment{True}

\ifdefined\ShowComment

\def\thatchaphol#1{\marginpar{$\leftarrow$\fbox{T}}\footnote{$\Rightarrow$~{\sf #1 --Thatchaphol}}}

\else

\def\thatchaphol#1{}

\fi

\ifdefined\ShowComment

\def\laszlo#1{\marginpar{$\leftarrow$\fbox{L}}\footnote{$\Rightarrow$~{\sf #1 --LK}}}

\else

\def\laszlo#1{}

\fi

\title{Binary search trees and rectangulations 
} 
\author{L\'{a}szl\'{o} Kozma\thanks{Saarland University, Saarbr\"{u}cken, Germany. Email: \texttt{kozma@cs.uni-saarland.de}} \quad \quad Thatchaphol Saranurak\thanks{KTH Royal Institute of Technology, Stockholm, Sweden. Email: \texttt{thasar@kth.se}}}

\begin{document}

%\begin{titlepage}
\maketitle

\begin{abstract}
We revisit the classical problem of searching in a binary search tree (BST) using rotations, and present novel connections of this problem to a number of geometric and combinatorial structures. In particular, we show that the execution trace of a BST that serves a sequence of queries is in close correspondence with the flip-sequence between two rectangulations. (Rectangulations are well-studied combinatorial objects also known as mosaic floorplans.) We also reinterpret \emph{Small Manhattan Network}, a problem with known connections to the BST problem, in terms of flips in rectangulations. We apply further transformations to the obtained geometric model, to arrive at a particularly simple view of the BST problem that resembles sequences of edge-relaxations in a shortest path algorithm.

Our connections yield new results and observations for all structures concerned. In this draft we present some preliminary findings. 
BSTs with rotations are among the most fundamental and most thoroughly studied objects in computer science, nonetheless they pose long-standing open questions, such as the dynamic optimality conjecture of Sleator and Tarjan (STOC 1983). 
Our hope is that the correspondences presented in this paper provide a new perspective on this old problem and bring new tools to the study of dynamic optimality.

\end{abstract}

\section{Introduction}

Binary search trees (BSTs) are among the simplest data structures for solving the \emph{dictionary problem} with keys from an ordered universe, supporting search, insert, delete, as well as other operations.\footnote{In this work we only focus on successful search operations, which we also call \emph{accesses}.} Alternatively, a BST can be seen as the implicit representation of a \emph{binary search} strategy for searching in an ordered list.

When a BST serves a \emph{sequence} of search queries, it is often advantageous to restructure the tree between queries (using rotations), paying a certain extra cost in the present in order to reduce the cost of queries in the future. Such a restructuring is called \emph{offline}, if it is done with advance knowledge of the entire query sequence, and it is called \emph{online}, if it may depend only on the queries already served.

Perhaps the best-known strategy for online BST re-arrangement is the Splay tree data structure of Sleator and Tarjan~\cite{ST85}. The cost of Splay was famously conjectured in 1983 to match the theoretical (offline) optimum on all inputs, up to a constant factor. An algorithm with this property is called \emph{constant-competitive}. The conjecture remains unresolved. An alternative, \emph{offline} algorithm was proposed by Lucas~\cite{Luc88} and independently by Munro~\cite{Mun00}, and conjectured to be constant-competitive. In a surprising development Demaine, Harmon, Iacono, Kane, and P\v{a}tra\c{s}cu (DHIKP)~\cite{DHIKP09} showed that the Lucas-Munro offline algorithm can be simulated by an online algorithm with only a constant factor slowdown. We refer to the resulting online algorithm of DHIKP simply as Greedy.

Several properties of Splay and Greedy are known, but we still seem far from proving constant-competitiveness for either of the two algorithms (or indeed, for any algorithm). Not only is it unknown whether an online algorithm can match the optimum, we also lack an efficient \emph{offline} algorithm for computing a good BST re-arrangement for a sequence of queries. 

Many of the recent results for the problem, including the development of the online Greedy algorithm, are based on a \emph{geometric view} of the BST model, introduced by DHIKP~\cite{DHIKP09}. (A somewhat similar model was described earlier by Derryberry, Sleator, and Wang~\cite{DSW}.) 

The elegance and usefulness of the geometric model lies in the fact that it hides the details of the tree re-arrangement (i.e.\ the concrete rotations that are performed), reducing the BST problem to a clean geometric optimization problem that requires finding a minimum \emph{satisfied superset} of a given point set in the plane. In this view, Greedy emerges as the most natural algorithm, equivalent to a simple geometric sweepline strategy. The geometric view seems less suitable for analysing algorithms other than Greedy (such as Splay, or Tango trees~\cite{Tango}). Furthermore, the property of a point set of being satisfied is \emph{non-monotone}, i.e.\ adding more points to a satisfied point set may destroy the property -- this unusual characteristic of the geometric optimization problem makes it difficult to apply standard algorithmic techniques to it. Informally, the difficulty in designing algorithms in the geometric view stems from the fact that Greedy is so natural, that it is unclear why any algorithm should deviate from it.

\paragraph{Rectangulations.}

The geometric model of DHIKP (described in \textsection\,\ref{sec2}) is the starting point of the work presented in this paper. We show a surprising equivalence of this geometric model with a well-studied and rich combinatorial structure. 

A \emph{rectangulation}\footnote{Alternative names include: rectangular subdivision, dissection, mosaic floorplan, or tesselation.} of an axis-parallel rectangle $R$ is a subdivision of $R$ into rectangles by axis-parallel line segments, no two of which may cross. A rectangulation is called \emph{slicing} (or \emph{guillotine}) if it can be obtained by recursively cutting a rectangle with a horizontal or vertical line into two smaller rectangles. See Figure~\ref{fig_r} for illustration.

The study of rectangulations is motivated by several applications. For geometric problems such as \emph{point location}, \emph{nearest neighbor}, or \emph{range searching}, the commonly used data structures rely on spatial subdivisions such as trapezoidations or rectangulations~\cite{comp_geom_book, seidel_planar_point}. The popular $k$-d tree corresponds (in the planar case) exactly to a slicing rectangulation~\cite{Bentley}. Rectangulations also appear in geometric approximation algorithms~\cite{Mit96}. In data visualization, ``cartograms'' based on rectangulations have been used for almost a century to represent both quantitive and relational information~\cite{eppstein_mumford, cartogram}. Rectangulations are also used to model problems in VLSI circuit design~\cite[\textsection\,53]{Handbook}. In communication complexity, a comparison protocol~\cite{kushilevitz} for a bivariate function $f$ corresponds to a slicing rectangulation in which every rectangle is $f$-monochromatic.

\begin{figure}[h]
\centering
\includegraphics[width=0.8\textwidth]{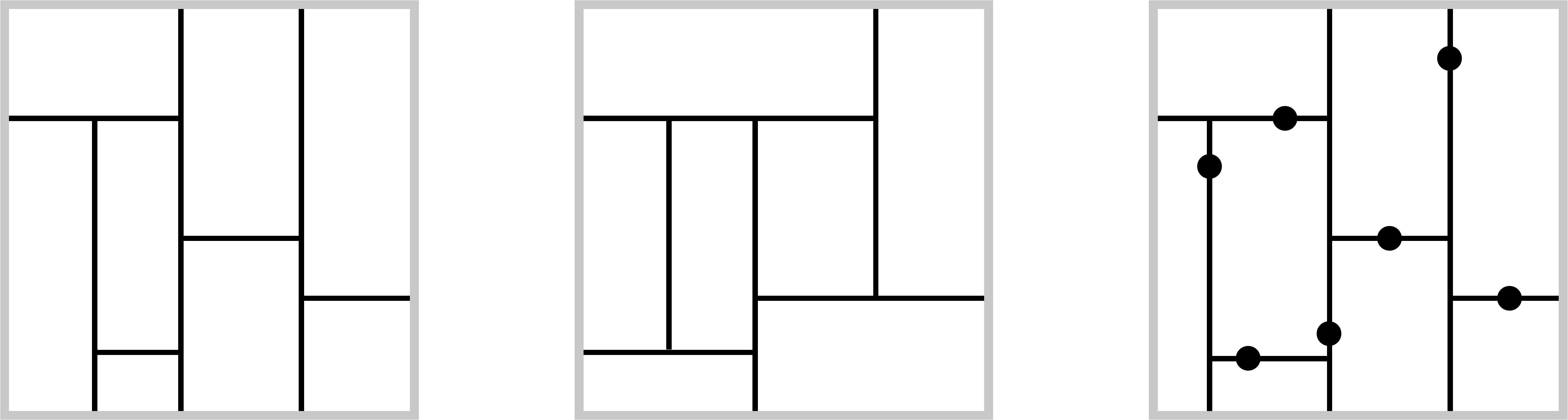}
\caption{Rectangulations. From left to right: (i) slicing rectangulation, (ii) non-slicing rectangulation, (iii) rectangulation constrained by points.\label{fig_r}}
\end{figure}

Several theoretical aspects of rectangulations have been studied in the combinatorics literature. In particular, it is known that the number of combinatorially different rectangulations with $n$ rectangles is given by the $n$th Baxter number~\cite{graham, baxter_oeis}, and the number of combinatorially different \emph{slicing} rectangulations with $n$ rectangles is given by the $n$th Schr\"oder number~\cite{graham, schroder_oeis}. Explicit bijections from rectangulations (general, respectively, slicing) have been given to natural classes of permutations counted by Baxter, respectively, Schr\"oder numbers~\cite{bijection}. Felsner~\cite{felsner} considers various ways in which rectangulations can represent certain classes of (planar) graphs.

Most relevant to our work is the recent paper of Ackerman et al.~\cite{Ackerman} that studies certain local operations (called \emph{flip} and \emph{rotate}) that transform one rectangulation into another. Following the definition of Ackerman et al.~\cite{Ackerman} we additionally constrain a rectangulation with a set $P$ of points (no two points on the same vertical or horizontal line), requiring that every point in $P$ is contained in the interior of a segment of the rectangulation (Figure~\ref{fig_r}).
Rectangulations constrained by points have also received attention in the literature (we refer to~\cite{Ackerman} and references therein). The flip and rotate operations in rectangulations were first introduced in~\cite{Fliporig}.

Ackerman et al.~\cite{Ackerman} study sequences of rectangulations constrained by the same set $P$ of points. In particular, they are interested in the \emph{flip diameter} of rectangulations, i.e.\ the maximum number of flip and rotate operations that may be required to transform one rectangulation into another. 

\paragraph{Our results.}
Our main result is that the problem of finding short sequences of flips between rectangulations is, in a precise sense, equivalent with the problem of finding short sequences of rotations in binary search trees for serving a sequence of search queries (i.e.\ the standard BST problem). We refer to \textsection\,\ref{sec2} for a precise statement of results. 

The connection between the BST problem and rectangulations is as follows. The sequence of search queries for a BST is mapped to a set of points in the plane -- these are the points constraining the rectangulations. The source and target rectangulations are the two that are (intuitively) furthest apart: the rectangulation consisting of only vertical lines and the rectangulation consisting of only horizontal lines.
Every sequence of flips that transforms the source rectangulation into the target rectangulation encodes the execution trace of a BST algorithm serving the given sequence of queries. Conversely, every sequence of BST re-arrangements that serves the query sequence encodes a valid sequence of flips from the all-vertical to the all-horizontal rectangulation.

This new ``flip-view'' of the BST problem allows yet another reinterpretation of known concepts from the BST world. In particular, in this model, the Greedy algorithm appears as one of many natural strategies. 

In flip-view, the BST rotation sequence emerges in an order that is different from both the temporal ordering of the input, and the spatial ordering of the keys -- the flip-sequence ``constructs'' the BST solution in an order that is constrained by the internal structure of the query-sequence. Since the flip-sequence has a clear goal (the all-horizontal rectangulation), there is also a clear sense of progress, which (we hope) makes this view more amenable for algorithm-design. As the online/offline distinction is less obvious here than in the geometric view of DHIKP, we find it possible that our new model is most suitable for the design and analysis of \emph{offline} BST algorithms. We use this new view of the BST problem to make some preliminary observations.

We further transform the obtained model, to arrive at a particularly simple formulation of the BST problem, as a problem resembling edge relaxations in a shortest path tree. 
This view appears even more algorithm-friendly than the other, as it makes a certain recursive structure of the problem apparent.  
The interpretation of Greedy in this model is simple and natural. We believe this model to give additional insight about the Greedy algorithm, and we hope that it will facilitate both the analysis of Greedy, and the design of new algorithms. We briefly explore these topics in \textsection\,\ref{sec3}.

Our definition of the flip operation is slightly different from the definition of Ackerman et al.~\cite{Ackerman}. Nevertheless, there is a clear relation between the two definitions, which allows us to answer an open question raised by Ackerman et al.\ concerning the flip diameter of rectangulations. The equivalence also leads to a simplified proof of a result shown by Ackerman et al., which arises now as an immediate corollary of known results for the BST problem. We explore this topic in \textsection\,\ref{sec_rect}.

\paragraph{Small Manhattan Network.}
The problem of connecting a given set of points in the plane by a manhattan network is a classical network design problem that has received significant attention. The variant which we consider here was studied by Gudmundsson, Klein, Knauer, and Smid (GKKS)~\cite{manhattan}. In this problem, the goal is to minimize the number of additional points added to a point set, in order to connect all original points with manhattan paths. Such a construction is also known as an $L_1$-spanner. (We give more precise definitions in \textsection\,\ref{sec2}.) This problem has a known connection to the BST problem, described by Harmon~\cite{Harmon}: The size of the optimum Small Manhattan Network solution is a lower bound for the optimum cost of serving a BST access sequence. Moreover, this lower bound is efficiently computable (a constant-approximate solution for Small Manhattan Network can be computed in polynomial time~\cite{Harmon, DHIKP09}). 

This connection seems not widely known. In particular, using this connection, some of the results of GKKS arise as corollaries of known facts about BSTs. We make this connection explicit, and we formulate further properties of small manhattan networks, following from recent results for BSTs~\cite{FOCS15}. We also interpret small manhattan networks in the rectangulation flip-view outlined earlier. The gap between the BST optimum and the Small Manhattan Network optimum is mysterious -- it is a long-standing conjecture that the two quantities are asymptotically the same. We are not yet able to settle this conjecture, but we believe that our new model of the manhattan network problem gives additional insight about its relation with the BST problem. We explore this topic in \textsection\,\ref{sgr} and \textsection\,\ref{sec_manh}.

\paragraph{Further related work.}

Binary search trees are counted by the Catalan numbers, therefore, they are in bijection with (literally) hundreds of known combinatorial structures that are similarly counted by the Catalan numbers~\cite{stanley_book}. The results of our current paper are (as far as we see) unrelated to these correspondences -- instead of a single BST, we study \emph{sequences of rotations} in a BST that serve a given sequence of queries. 
  
There exist known connections between rectangulations and BSTs. In particular, slicing rectangulations have a straightforward BST-representation, which is useful in geometric applications such as planar point location. For general rectagulations more complex BST-based representations are known, such as the twin binary tree structure given by Yao, Chen, Cheng, and Graham~\cite{graham}. The connection described in our paper is, again, very different from such results. We do, in fact, relate a sequence of rectangulations with a sequence of BSTs. However, in our model, the intermediate elements in the two sequences are not in direct correspondence with each other. An intermediate rectangulation in our sequence corresponds to an abstract state of a BST algorithm, in which some partial structure of the intermediate trees has been committed to, while other structure is still left undecided. We find it an intriguing question, whether the known BST-based representations of rectangulations have any relevance to the connections introduced in our current paper.

\section{The main equivalences} \label{sec2}

\subsection{Binary search tree (BST) problem} Let $[n]=\{1,\dots,n\}$, and let $X = (x_1, \dots, x_m) \in [n]^m$ be an \emph{access sequence}. We view $X$ at the same time as a collection of points in the plane in a straightforward way: $X = \{(x_i,i) : 1 \leq i \leq m\}$. For any point $p$, we denote by $p.x$ and $p.y$ the $x$-coordinate and the $y$-coordinate of $p$ respectively.

A BST algorithm $\Ac$ reads the sequence $X$, and outputs an initial binary search tree with nodes $[n]$, and a sequence of operations \emph{moveleft}, \emph{moveright}, \emph{moveup}, \emph{rotate}. 
We say that $\Ac$ \emph{serves} $X$, if the initial tree and the sequence of operations encode a valid sequence of pointer-moves and rotations-at-the-pointer, such that each element $x_i$ is in turn moved to the root. The pointer starts at the root of the initial tree. The \emph{cost} of $\Ac$ serving $X$, denoted $\Ac(X)$ is the number of operations output by $\Ac$. 
This cost model follows the description of Wilber~\cite{Wilber}, and is equivalent within a constant factor with several other descriptions of the BST model. In particular, it is easy to see that requiring the accessed element to become the root entails only a constant factor change in the cost.

We say that $\Ac$ is \emph{offline}, if it has access to the entire input $X$ at once, and we say that $\Ac$ is \emph{online}, if it reads the input $X$ one element at a time. After reading $x_i$, for all $i$, an online algorithm outputs a sequence of operations that bring  $x_i$ to the root. 

In the following, $X$ is always a permutation, i.e.\ $m=n$, and $x_i \neq x_j$, for all $i,j$. We call the corresponding point set a \emph{permutation point set}. It is known~\cite{DHIKP09} that an offline BST algorithm $\Ac$ that serves only permutations as input can be transformed with a constant factor slowdown into an offline algorithm $\Ac'$ that serves arbitrary sequences as input. (The statement also holds for online algorithms~\cite{FOCS15} under mild conditions on $\Ac$.)

\subsection{Satisfied Superset problem} \label{sec:ssp} This problem is defined by DHIKP~\cite{DHIKP09}. A point set $Y \in [n] \times [n]$ is \emph{satisfied}\footnote{The term used by DHIKP is \emph{arborally satisfied}.}, if for any two points $a,b \in Y$, one of the following holds: (i) $a$ and $b$ are on the same horizontal or vertical line, or (ii) the rectangle with corners $a$ and $b$ contains some point in $Y \setminus \{a,b\}$, possibly on the boundary of the rectangle.

Given a permutation point set $X$ of size $n$, an algorithm $\Ac$ for the Satisfied Superset problem outputs a point set $Y$, with $X \subseteq Y \subseteq [n] \times [n]$, such that $Y$ is satisfied. The cost of $\Ac$ is the size of the set $Y$, denoted $\Ac(X)$.

\begin{theorem}[DHIKP~\cite{DHIKP09}] \label{lem1}
Any algorithm $\Ac$ for the Satisfied Superset problem can be transformed (in polynomial time) into an algorithm $\Ac'$ for the BST problem, such that for all inputs $X$ we have $\Ac'(X) = \Theta\big(\Ac(X)\big)$. Furthermore, any algorithm $\Ac'$ for the BST problem can be transformed (in polynomial time) into an algorithm $\Ac$ for the Satisfied Superset problem, such that for all inputs $X$ we have $\Ac(X) = \Theta\big(\Ac'(X)\big)$.
\end{theorem}

More strongly, DHIKP show that there is a one-to-one correspondence between the points in the Satisfied Superset solution and the nodes of the tree that are touched by rotations at any given time in the BST solution.

A \emph{manhattan path} of length $k$ between two points $x,y \in A$ with respect to $B$, where $B \supseteq A$, is a sequence of distinct points $(x=x_1, x_2, \dots, x_k=y) \in B^k$, such that for all $i = 1,\dots,k-1$ the two neighboring points $x_i, x_{i+1}$ are on the same horizontal or vertical line, %the segment $[x_i, x_{i+1}]$ contains no other points from $B$, 
and both the $x$-coordinates and the $y$-coordinates of $(x_1,\dots,x_k)$ form a monotone sequence.

An alternative definition of a \emph{satisfied point set} given by Harmon~\cite{Harmon} is the following: 
\begin{proposition}[Harmon~\cite{Harmon}] \label{prop: sat and manhattan}
	$Y \in [n] \times [n]$ is satisfied, if for any two points $a,b \in Y$, there is a manhattan path between $a$ and $b$ with respect to $Y$.
\end{proposition}
Verifying the equivalence of this definition with the previous one is an easy exercise. 

\newpage
\subsection{Rectangulation problem} \label{sec_r_intro}
Variants of this problem have been studied in the literature. The formulation we describe here is new, but closely related to the problem studied by Ackerman et al.~\cite{Ackerman}. The exact difference between our model and the model of Ackerman et al., and the implications of this difference are explored in \textsection\,\ref{sec_rect}. 

Let $n$ be an arbitrary integer (the problem size). We define the set of planar points $S = \{0,1,\dots,n+1\} \times \{0,1,\dots,n+1\}$. Points in $C = \big\{ (0,0)$, $(0,n+1)$, $(n+1,0)$, $(n+1,n+1) \big\}$ are called \emph{corner points} and will not be used in any way. Points in $M = \big\{(i,0), (i,n+1), (0,i), (n+1,i) : i \in [n] \big\}$ are called \emph{margin points}.  The remaining points (i.e.\ those in $S \setminus (M \cup C) = [n] \times [n]$) are called \emph{non-margin points}.

A \emph{state} $(P,L)$ of the Rectangulation problem consists of a set $P \subseteq (S \setminus C)$ of points and a set $L$ of horizontal and vertical line segments (in the following, simply segments) with endpoints in $P$. \\

A state $(P,L)$ is \emph{valid} iff it fulfills the following conditions (see Figure~\ref{fig:flip_view}): 
\begin{compactenum}[(i)]
\item Each segment in $L$ contains exactly two points from $P$, namely its two endpoints. \\ (This implies that no point from $P$ is in the interior of a segment in $L$.)
\item No two segments in $L$ intersect each other (except possibly at endpoints). 
\item $(P,L)$ is $\emph{elbow-free}$. This means that each non-margin point in $P$ is contained in at least two segments of $L$, and if it is contained in exactly two segments, then they must have the same orientation (i.e.\ either both vertical or both horizontal).
\end{compactenum}\ \\

\vspace{-0.2in}

\noindent The \emph{initial state} $(P_0, L_0)$ is defined with respect to an input permutation point set $X$ of size $n$. The set $P_0$ is equal to $X \cup M$. The set $L_0$ contains for each non-margin point $(x,y) \in P_0 \setminus M$ two vertical segments: the one between $(x,0)$ and $(x,y)$, and the one between $(x,y)$ and $(x,n+1)$. It is easy to see that the initial state is valid.

An \emph{end state} $(P^*, L^*)$ is a valid state that consists of a point set $P^* \supseteq P_0$, and a set of segments $L^*$, all of them horizontal, such that they cover every point $\{0,1,\dots,n+1\} \times [n]$. See Figure~\ref{fig:flip_view} for illustration.

Given a permutation point set $X$, an algorithm $\Ac$ for the Rectangulation problem transforms the initial state $(P_0,L_0)$ determined by $X$ into an end state $(P^*, L^*)$, through a sequence of \emph{valid flips}, defined below. The cost of the algorithm, denoted $\Ac(X)$, is the number of flips in this sequence. In words, the goal is to go from the all-vertical state to the all-horizontal state through a minimum number of valid flips.

Let $(P,L)$ be a valid state. Two points $a,b \in S \setminus C$ define a flip, denoted $\left<a,b\right>$. A flip $\left<a,b\right>$ transforms the state $(P,L)$ into a new state $(P',L')$ as follows. First, we let $P' = P \cup \{a,b\}$, and $L' = L \cup \{ [a,b] \}$. If there exists some segment $[x,y] \in L$ that contains $a$ in its interior, we remove $[x,y]$ from $L'$, and add the segments $[x,a]$, and $[a,y]$ to $L'$. Similarly, if there exists some segment $[z,t] \in L$ that contains $b$ in its interior, we remove $[z,t]$ from $L'$, and add the segments $[z,b]$, and $[b,t]$ to $L'$.

For $\left<a,b\right>$ to be a valid flip, it must hold that the resulting state $(P',L')$ is a valid state. In particular, we can only add a segment $[a,b]$ if it is horizontal or vertical, and if it does not intersect existing segments (except at $a$ or $b$). 
After every flip we can remove from $L'$ an arbitrary number of segments. By removing segments we must not violate the elbow-free property. For instance, we can only remove a vertical segment if its non-margin endpoints are contained in two horizontal segments of $L'$ (in other words, the endpoints have extensions both to the left and to the right).

\begin{figure}[h]
\centering
\includegraphics[width=0.95\textwidth]{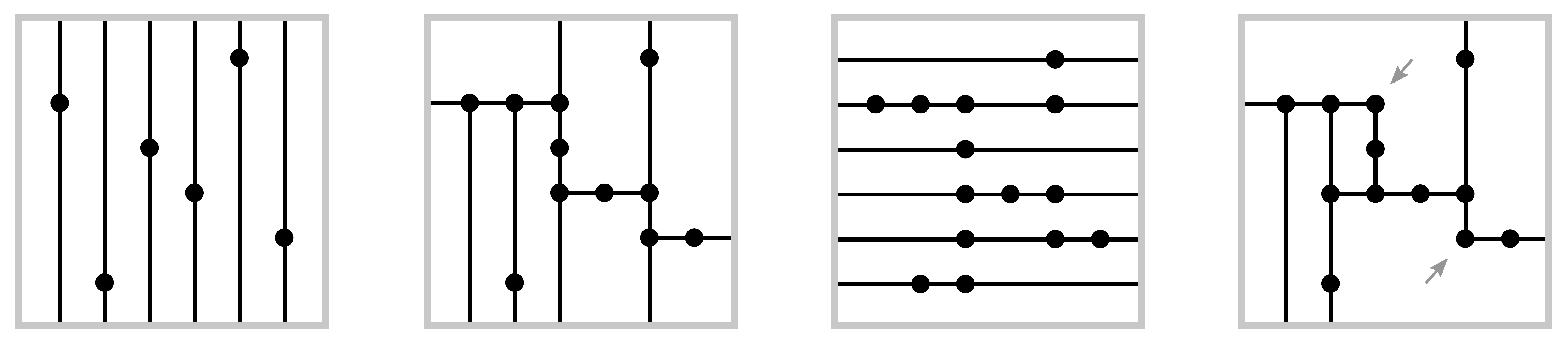}
\caption{Rectangulation problem. From left to right: (i) initial state corresponding to $X = (2,6,4,3,1,5)$, (ii) a valid intermediate state, (iii) a valid end state, (iv) an invalid intermediate state (observe that the state is not elbow-free). Margin points are not shown.\label{fig:flip_view}}
\end{figure}

\begin{figure}[h]
\centering
\includegraphics[width=0.6\textwidth]{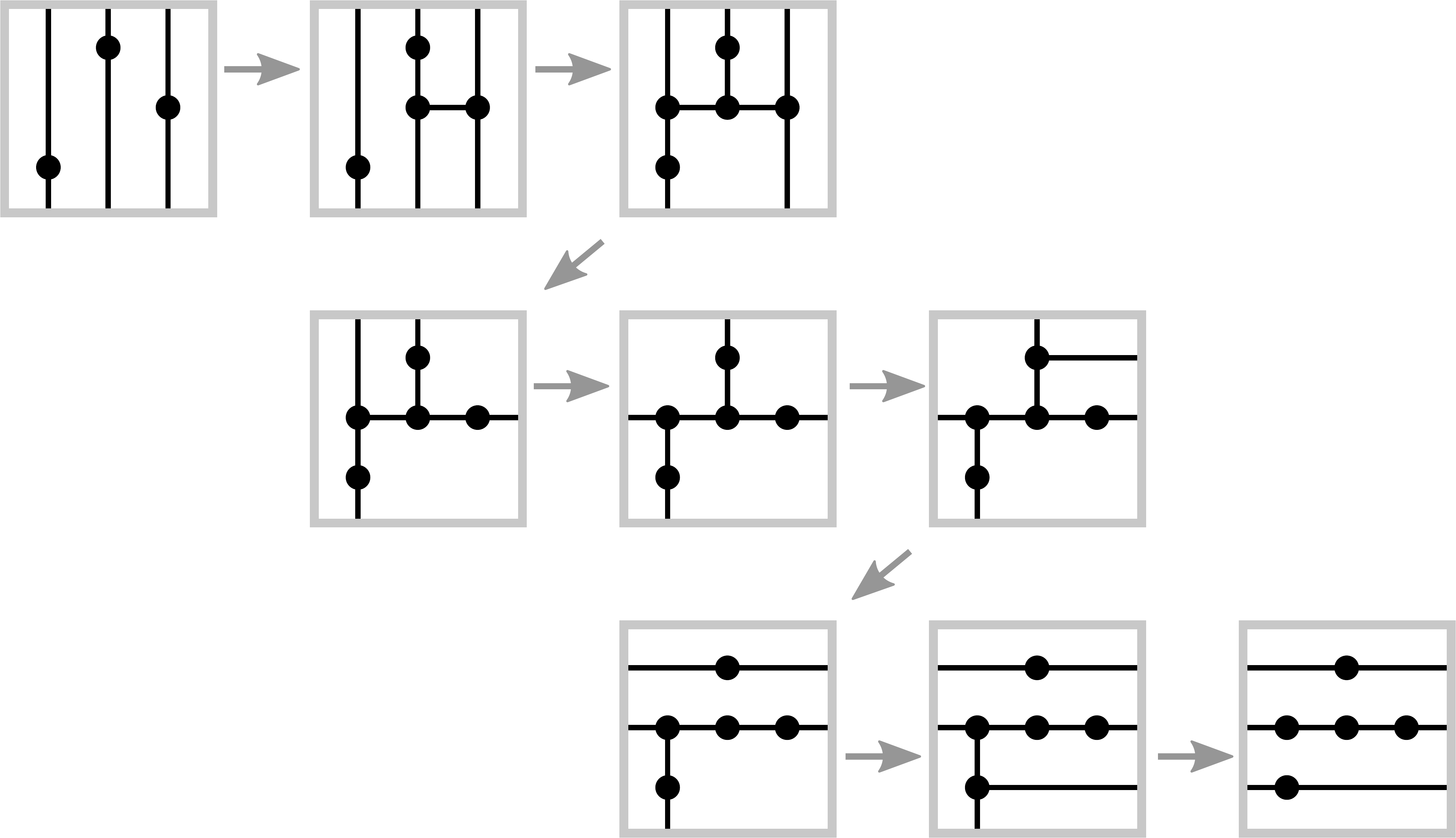}
\caption{A sequence of valid flips from initial state to end state. Margin points are not shown. \label{fig:flips}}
\end{figure}

The following theorem captures the connection (in one direction) between Rectangulation and Satisfied Superset. 

\begin{theorem}\label{thm1}
Any algorithm $\Ac$ for the Rectangulation problem can be transformed (in polynomial time) into an algorithm $\Ac'$ for the Satisfied Superset problem, such that for all inputs $X$, we have $\Ac'(X) = O\big(\Ac(X)\big)$.
\end{theorem}

\begin{proof}
Consider an algorithm $\Ac$ for Rectangulation, executed from initial state $(P_0, L_0)$, defined by an input permutation $X$. As we run $\Ac$, we construct a set $Y \supseteq X$, that is a solution for Satisfied Superset (this is our new algorithm $\Ac'$). The process is straightforward: Initially we let $Y=X$. Whenever $\Ac$ performs a flip $\left<a,b\right>$, we let $Y = Y \cup (\{a,b\} \setminus M)$. In words, we construct a superset of $X$ by adding every non-margin endpoint created while flipping from the all-vertical to the all-horizontal state in Rectangulation. The cost of $\Ac$ is equal to the number of flips. Since each flip adds at most two points to $Y$, the claim on the cost of $\Ac'$ is immediate.

It remains to show that the point set $Y$ thus constructed is satisfied. Suppose otherwise, that in the end there are two points $a,b \in Y$, that are not on the same horizontal or vertical line, and the rectangle with corners $a,b$ contains no other point of $Y$. Without loss of generality, assume that $a$ is above and to the left of $b$. Let $\left<a,a'\right>$ be the last flip in the execution of $\Ac$ such that $a'$ is on the same horizontal line as $a$ and to the right of $a$. Let $\left<b',b\right>$ be the last flip such that $b'$ is on the same horizontal line as $b$ and to the left of $b$. %$b$ be the first flip of this type
(There have to be such flips, otherwise $\Ac$ would not produce a valid end state.) Since the rectangle with corners $a,b$ is empty, $b'$ must be to the left of $a$, and $a'$ must be to the right of $b$. 

%Observe that there must exist a point $a''$ strictly to the left of $a$, such that the horizontal segment $[a'',a]$ is in $L^*$. 
Suppose that the flip $\left<b',b\right>$ occurs earlier than the flip $\left<a,a'\right>$ (the other case is symmetric), and consider the state before the flip $\left<a,a'\right>$. In that state there must be a vertical segment with top endpoint at $a$, otherwise $a$ would be contained in at most two segments, not both horizontal or vertical, contradicting the elbow-free property. %(Observe that once created, a point or a horizontal segment is never removed from the state, hence, at this point we know that there is no segment with $a$ as its left endpoint.) 
Let $a^*$ be the bottom endpoint of the vertical segment with top endpoint $a$. The point $a^*$ must be strictly below $b$, for otherwise the rectangle with corners $a,b$ would contain it. This means that $[a,a^*]$ intersects $[b',b]$, contradicting that we are in a valid state. %A symmetric argument works in the case when $[a,a']$ is added earlier than $[b', b]$.
We conclude that $Y$ is a satisfied superset of $X$. \qedd
\end{proof}

The following converse of Theorem~\ref{thm1} also holds.

\begin{theorem}\label{thm2}
Any algorithm $\Ac'$ for the Satisfied Superset problem can be transformed (in polynomial time) into an algorithm $\Ac$ for the Rectangulation problem, such that for all inputs $X$, we have $\Ac(X) = O\big(\Ac'(X)\big)$.
\end{theorem}

\begin{proof}
Consider an algorithm $\Ac'$ for Satisfied Superset that for input $X$ outputs a point set $Y \supseteq X$. We construct a sequence of flips that transform the initial state $(P_0, L_0)$ of the Rectangulation problem determined by $X$ into an end state $(P^*, L^*)$, such that $P^* \setminus M = Y$ (this is our new algorithm $\Ac$). We define $\Ac$ such that every flip creates a new horizontal segment whose endpoints are in $Y \cup M$, and no horizontal segment is ever removed during the course of the algorithm.
The claim on the cost of $\Ac$ is immediate, since each flip can be charged to one of its (non-margin) endpoints, and each point in $Y$ has at most two flips charged to it. The removal of vertical segments does not contribute to the cost.

We run algorithm $\Ac$ until we reach an end state, maintaining the invariant that in every state $(P,L)$, we have $(P \setminus M) \subseteq Y$. The invariant clearly holds in the initial state $(P_0, L_0)$ determined by $X$. Algorithm $\Ac$ consists of two types of greedy steps, executed in any order:
(1) if at any point, some valid flip $\left<a,b\right>$ is possible, such that $a,b \in Y \cup M$, then execute it, and (2) if at any point, some vertical segment $[a,b] \in L$ that contains no point from $Y$ (except possibly its endpoints) can be removed, then remove it. %%Since the total number of possible steps of the above types is polynomial, the claim on the running time follows.

It remains to be shown that the algorithm does not get stuck, i.e.\ that there is always an operation of type (1) or (2) that can be executed, unless we have reached a valid end state. Consider an intermediate state $(P,L)$ during the execution of $\Ac$ and suppose for contradiction that there is no available operation of either type.

Consider two points $q, q' \in Y \cup M$ on the same horizontal line, $q$ to the left of $q'$, such that $[q,q']$ is not in $L$, and the segment $[q,q']$ contains no point of $Y$ in its interior. If there is no such pair of points, then we are done, since all horizontal lines are complete, and all remaining vertical segments can be removed. Among such pairs, consider the one where $q$ is the rightmost, in case of a tie, choose the one where $q'$ is the leftmost.

Call a point $q \in P$ \emph{left-extensible} if it is not the right endpoint of a segment in $L$, and \emph{right-extensible} if it is not the left endpoint of a segment in $L$. %A point $q \in P$ is \emph{complete} if it is neither left-extensible nor right-extensible. 

Observe that throughout the execution of $\Ac$, for any state $(P,L)$, every point in $Y$ is contained in some segment of $L$. Since $\left<q,q'\right>$ is not a valid flip, $[q,q']$ must intersect some vertical line $[z,z'] \in L$ (assume w.l.o.g.\ that $z$ is strictly above, and $z'$ is strictly below $[q,q']$). Observe that $[z,z']$ cannot contain a point of $Y$ in its interior. If it would contain such a point $z^*$, then $z^*$ would be the left endpoint of some segment missing from $L$, contradicting the choice of $q$. Thus, since removing $[z,z']$ is not a valid step, it must be that one of $z$ and $z'$ is a non-margin point that is left- or right-extensible.
If $z$ or $z'$ were right-extensible, that would contradict the choice of $q$. Therefore, one of them must be left-extensible, and assume w.l.o.g.\ that $z$ is left-extensible.

Since $Y$ is satisfied, by \Cref{prop: sat and manhattan} there has to be a point $w \in Y \setminus \{z,q\}$ either on the horizontal segment $[(q.x,z.y),z]$, or on the vertical segment $[z,(z.x,q.y)]$. 
Since $[z,z']$ cannot contain a point of $Y$ in its interior, it must be the case that $w$ is on $[(q.x,z.y),z]$, and choose $w$ to be closest to $z$. But then the segment $[w,z]$ is missing from $L$, contradicting the choice of $q,q'$ because $w.x \geq q.x$. \qedd
\end{proof}

Theorem~\ref{thm1} and Theorem~\ref{thm2} state that the Rectangulation and Satisfied Superset problems are polynomial-time equivalent. Observe that the proofs, in fact, show something stronger: For an arbitrary permutation point set $X$, a point set $Y \supseteq X$ is a solution for Satisfied Superset exactly if $Y \cup M$ is the point set of a valid (and reachable) end state for Rectangulation.

\subsection{Tree Relaxation problem} \label{sec_tr}

Consider again a permutation point set $X$ of size $n$ as input. A \emph{monotone tree} on $X$ is a rooted tree that has $X$ as the set of vertices, and whose edges are all going away from the root according to the vertical ordering of the points. That is, if two points $a = (a.x, a.y)$, $b = (b.x, b.y)$, with $a,b \in X$ are the endpoints of an edge in a monotone tree on $X$, then $a$ is closer to the root than $b$ (in graph-distance) iff $a.y < b.y$. Recall that all points in $X$ have distinct $x$- and $y$-coordinates. It follows that $x_1$ is the root of every monotone tree on $X$.

We are concerned with two special monotone trees on $X$. The \emph{treap} on $X$ is the binary search tree with the $x$-coordinates as keys, and the $y$-coordinates as heap-priorities. That is, the lowest point $x_1 \in X$ is the root of the tree, and the points left of $x_1$ form its left subtree, and the points right of $x_1$ form its right subtree, defined in a recursive fashion. (We refer to~\cite{SeidelA96} for results on treaps.) The \emph{path} on $X$ is a tree that connects all elements through a path by increasing $y$-coordinate, i.e.\ in the order $x_1, \dots, x_n$. It is easy to verify that both the treap and the path defined on $X$ are unique and that they form monotone trees on $X$. Observe that the definition of a monotone tree does not require the tree to fulfill the search tree property or even to be binary. See Figure~\ref{fig_tree} for illustration.

Given a permutation point set $X$, an algorithm $\Ac$ for the Tree Relaxation problem transforms the treap on $X$ to the path on $X$ through a sequence of \emph{valid edge-flips}, defined below. The cost of the algorithm, denoted $\Ac(X)$, is the number of edge-flips in this sequence.

Let $T$ be a monotone tree on $X$. A valid edge-flip in $T$ is defined as follows. Consider a vertex $r$ of $T$ that has at least two children. Sort the children of $r$ by their $x$-coordinate, and let $a$ and $b$ be two children that are neighbors in this sorted order, such that $a$ is below $b$ (the $y$-coordinate of $a$ is smaller than the $y$-coordinate of $b$). Then the edge-flip $(a \rightarrow b)$ adds the edge $(a,b)$ to $T$ and removes the edge $(r,b)$ from $T$. It is easy to verify that a valid edge-flip maintains the monotone tree property of $T$. The edge-flip operation is reminiscent of an edge-relaxation in shortest-path algorithms (performed in reverse). See Figure~\ref{fig_tree2} for illustration.

\begin{figure}[h]
\centering
\includegraphics[width=0.8\textwidth]{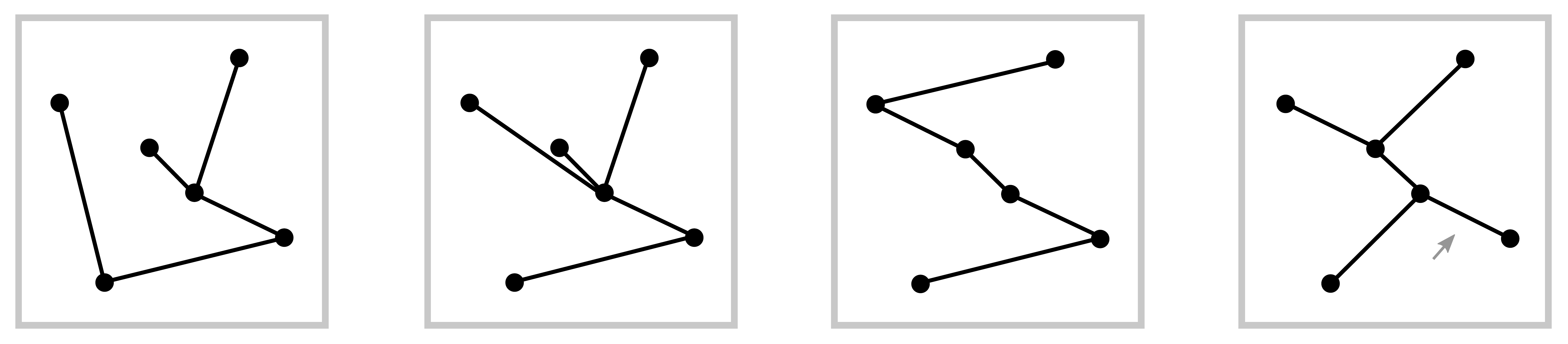}
\caption{Tree relaxation problem. From left to right: (i) treap on $X = (2,6,4,3,1,5)$, (ii) an intermediate monotone tree on $X$, (iii) path on $X$, (iv) an invalid intermediate state (tree is not monotone).\label{fig_tree}}
\end{figure}

\begin{figure}[h]
\centering
\includegraphics[width=0.8\textwidth]{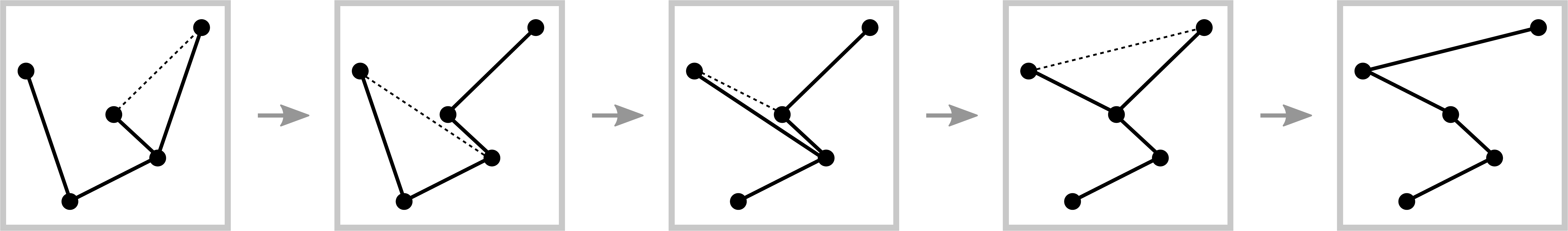}
\caption{A sequence of valid edge-flips from initial (treap) state to end (path) state. \label{fig_tree2}}
\end{figure}

The Tree Relaxation problem is closely related to the Rectangulation problem (and as a consequence, to the BST problem), as shown by the following theorem. 

\begin{theorem} \label{tree_rect}
Any algorithm $\Ac$ for the Tree Relaxation problem can be transformed (in polynomial time) into an algorithm $\Ac'$ for the Rectangulation problem, such that for all inputs $X$ of size $n$, we have $\Ac'(X) = O(\Ac(X) + n)$.

\end{theorem}

\begin{proof}
We start $\Ac'$ with an \emph{initial phase}, then we simultaneously run algorithm $\Ac$ for Tree Relaxation on $X$, and output the operations of $\Ac'$ for Rectangulation on $X$, such that we output at most two flips in $\Ac'$ for every edge-flip in $\Ac$. 
We finish $\Ac'$ with a \emph{cleanup phase}. Both the initial phase and the cleanup phase consist of $O(n)$ flips, to be specified later.

In any given state $(P,L)$ during the execution of $\Ac'$, let $H_i$ denote the set of horizontal segments in $L$ at height $i$. After the initial phase, we maintain throughout the execution of $\Ac'$ the following invariants, denoted $I_1$, $I_2$, and $I_3$. \\

\textbf{$I_1$ (contiguity):} For all $i$, the union of the horizontal segments in $H_i$ form a contiguous horizontal segment, which we denote $h_i$. \\

\textbf{$I_2$ (nesting):} 
%for all $i$ and $j$, such that $j>i$, $h_i$ and $h_j$ are \emph{nested}. 
For all $i$, denote the $x$-coordinate of the left (resp.\ right) endpoints of $h_i$ as $\ell_i$ (resp.\ $r_i$). Let $x_{i_1}, \dots, x_{i_k}$ be the children of $x_t$ in the current monotone tree on $X$, sorted by $x$-coordinate (i.e.\ $x_{i_1} < \cdots < x_{i_k}$). We have that the endpoints of the segments $h_{i_1}, \dots, h_{i_k}$ are aligned, and not overhanging the parent segment $h_t$. More precisely, $\ell_t \leq \ell_{i_1}$, and $r_{i_k} \leq r_t$, and for all $j=1,\dots,k-1$, we have $r_{i_j} = \ell_{i_{j+1}}$. For the root $x_1$ of the tree, we have $\ell_1 = 0$, and $r_1 = n+1$.\\

\textbf{$I_3$ (visibility):} Let $x_{i_1}, \dots, x_{i_k}, x_t$ be defined as before. % $i$ and $t$ such that $x_i$ is the child of $x_t$ in the current monotone tree on $X$,
For $j=1,\dots,k$, let us denote by $\Rc_j$ the axis-aligned rectangle with corners $(\ell_{i_j}, i_j)$, $(r_{i_j}, t)$. Let $\Rc_0$ be the rectangle with corners $(\ell_t, t)$, $(\ell_{i_1}, n+1)$, and let $\Rc_{j+1}$ be the rectangle with corners $(r_{i_k},n+1)$, $(r_t, t)$. We have that the interiors of the rectangles $\Rc_0, \Rc_1, \dots, \Rc_{j+1}$ are not intersected by any segment in $L$ in the current state of $\Ac'$. Furthermore, the vertical sides of the rectangles $\Rc_0, \Rc_1, \dots, \Rc_{j+1}$ are either touching the margin, or fully covered by segments in $L$. \\

In the end we show that $\Ac'$ reaches a valid end state of Rectangulation. Since the total number of flips performed is at most $O(n) + 2 \cdot \Ac(X)$, the claim on the cost follows. It remains to describe the steps of the algorithm $\Ac'$.

\paragraph{Initial phase.}

For each $(i = n, \dots, 1)$, flip $\left<L_i, x_i\right>$ and $\left<x_i, R_i\right>$, where $L_i$ ($R_i$) is the leftmost (rightmost) point such that the corresponding flip is valid. After every flip, remove all possible vertical segments from the current state before proceeding to the next $i$. In particular, remove the vertical segment with endpoints $(x_i,0)$, $(x_i,i)$. \\

Recall that in the beginning, $\Ac$ is in the state that the current tree is the treap on $X$. Let $T$ be the treap on $X$. We show that the invariants hold after the initial phase. Observe that $I_1$ holds trivially: After the initial phase we have the horizontal contiguous segments $h_i = [L_i, R_i]$. 

We prove $I_2$ and $I_3$ by induction. Clearly, if $|X| = 1$, the invariants hold. Consider the last step, when we flip $\left<L_1, x_1\right>$ and $\left<x_1, R_1\right>$, for suitable $L_1$, $R_1$. Denote by $v$ the vertical segment with $x$-coordinate equal to $x_1$. Observe that for all $i>1$, the segments $h_i$ can intersect $v$ only at their endpoints ($L_i$ and $R_i$). Since none of these points are extended both to the left and to the right, no portion of $v$ has been removed before this step. This means that the rectangulations on the two sides of $v$ are independent of each other, i.e.\ they would have been the same even if the input on the other side of $v$ were different. In particular, this means that, by induction $I_2$ and $I_3$ hold for the rectangulations on the two sides of $v$, corresponding to the left and right subtrees of $x_1$ in the treap $T$. 

Let $x_i$ and $x_j$ be the left, respectively right child of $x_1$ in $T$ (one of the two might be missing, in case $x_1 = 1$ or $x_1 = n$). By $I_2$, we have that $h_i$ extends horizontally from $0$ to $x_1$, and $h_j$ extends horizontally from $x_1$ to $n+1$. Since the vertical segments below every point $x_i$ have been removed, we can execute the flips $\left<(0,1),x_1\right>$ and $\left<x_1,(n+1,1)\right>$. Finally, since $x_1$ is complete, we can remove the vertical segment with top endpoint $x_1$. %\thatchaphol{How do you use the induction hypothesis that $I_3$ holds?}

Both $I_2$ and $I_3$ are established for the full tree $T$, completing the induction.

\paragraph{Flips during the execution of $\Ac$.}

Let $x_i$ and $x_j$ be neighboring children of $x_t$ in the current tree, such that $i<j$, and the valid edge-flip $(x_i \rightarrow x_j)$ is executed in $\Ac$. Assume w.l.o.g.\ that $x_i < x_j$. By $I_1$, there exist contiguous horizontal segments $h_i$ and $h_j$.
Observe that $h_i$ is below $h_j$. Let $\ell_i$, $\ell_j$, and $\ell_t$ ($r_i$, $r_j$, and $r_t$) denote the $x$-coordinates of the left (right) endpoints of $h_i$, $h_j$, and $h_t$. By $I_2$, we have $\ell_t \leq \ell_i < r_i = \ell_j < r_j \leq r_t$. Let $L = (r_i, i)$, and let $R = (r_j, i)$. Then the flip $\left<L,R\right>$ is executed in $\Ac'$. 

Let us verify that the flip is valid. Due to invariant $I_3$, the rectangle with corners $(r_i,i)$ and $(r_j,j)$ has empty interior, therefore the flip intersects no vertical segment. Moreover, by $I_3$, the right side of the rectangle is covered by segments. Therefore, the flip creates no crossing, elbow, or point of degree one, it is therefore a valid flip. 

After the flip, $x_j$ is the child of $x_i$, and $x_i$ is the child of $x_t$. It is easy to verify that the invariants are maintained, except for the following case: Let $x_k$ be the rightmost child of $x_i$ \emph{before the flip}, with endpoints $\ell_k$ and $r_k$. After the flip, $x_k$ and $x_j$ are neighboring siblings, but it may happen that their endpoints are not aligned, i.e.\ $r_k < \ell_j$. If this is the case, we need to perform an additional flip. Suppose that $h_k$ is lower than $h_j$. Then execute in $\Ac'$ the flip $\left<(r_k, j), (\ell_j, j) \right>$. In the case when $h_j$ is lower than $h_k$, execute the flip $\left<(r_k, k), (\ell_j, k) \right>$. The flips are valid by $I_3$ before the edge-flip in $\Ac$, and by the flip, $I_3$ is re-established using at most two flips. % are Finally, we remove a maximal possible number of vertical segments. It is easy to verify that all three invariants are maintained. 

In the end, if $\Ac$ is a correct algorithm for Tree Relaxation, it will end with a path tree. Let $h_1, \dots, h_n$ be the horizontal lines corresponding to the current state in the execution of $\Ac'$. 

\paragraph{Cleanup phase.} From invariants $I_1$, $I_2$, and $I_3$, it follows that $\ell_1 \leq \ell_2 \leq \dots \leq \ell_n < r_n \leq \dots \leq r_1$. We can transform this state to a valid end state for Rectangulation, with the valid flips (in this order): $\left<0,\ell_1\right>, \left<0, \ell_2\right>, \dots, \left<0, \ell_n\right>, \left<r_1, n+1\right>, \dots, \left<r_n, n+1\right>$. \qedd

\end{proof}

\subsection{Signed Satisfied Superset problem} \label{sgr}

\input{signed}

\section{Consequences for Flip Diameter}\label{sec_rect}

In this section, we study the Flip Diameter problem introduced by Ackerman et al.~\cite{Ackerman}. Ackerman et al.~study the distance between two rectangulations constrained by the same set of points, where distance refers to the shortest sequence of local operations that transform one rectangulation into the other. 

The concept of rectangulation studied by Ackerman et al.\ is the same as the one we defined in~\textsection\,\ref{sec_r_intro}, apart from the fact that we keep track of all intersection points that are created in the sequence of transformations from one rectangulation to another, whereas in the problem studied by Ackerman et al.\ this is not explicitly needed. (Their definition of a rectangulation is essentially the union of all segments in $L$, for a given state $(P,L)$.) 

The Flip Diameter problem asks, given a set of points constraining rectangulations, to find the largest possible distance between two rectangulations. This is in contrast to the problem described in~\textsection\,\ref{sec_r_intro}, where we are concerned with the distance between two \emph{particular} rectangulations, namely the all-horizontal, and the all-vertical one.

Finally, the local operations used by Ackerman et al.\ are slightly different from the flip operation we define in~\textsection\,\ref{sec_r_intro}. In the following, we define the rotate and flip operation used by Ackerman et al.\ in the context of our Rectangulation problem. We call these two operations \emph{A-rotate} and \emph{A-flip}.

Given a valid state $(P,L)$ of Rectangulation, an \emph{A-rotate} operation consists of removing a segment $[x,y]$ from $L$, and adding a new segment $[x,z]$ to $L$. If $z$ is contained in the interior of some segment $[a,b] \in L$, we remove $[a,b]$ from $L$, and add $[a,z]$ and $[z,b]$ to $L$. We denote the resulting set of segments $L'$ and we let $P' = P \cup \{z\}$. The A-rotate operation is valid, if $[x,y]$ and $[x,z]$ have different orientations (i.e.\ one of them horizontal, the other vertical), and if the resulting state $(P',L')$ is a valid state of Rectangulation.

Given a valid state $(P,L)$ of Rectangulation, an \emph{A-flip} operation consists of removing two segments $[x,y]$ and $[y,z]$ from $L$ and adding new segments $[v,y]$ and $[y,w]$ to $L$. If $v$ is contained in the interior of some segment $[a,b] \in L$, we remove $[a,b]$ from $L$ and add $[a,v]$ and $[v,b]$. Similarly, if $w$ is contained in the interior of some segment $[c,d] \in L$, we remove $[c,d]$ from $L$ and add $[c,w]$ and $[w,d]$. We denote the resulting set of segments $L'$ and we let $P' = P \cup \{v,w\}$. The A-flip operation is valid, if $[x,y]$ and $[y,z]$ have the same orientiation (i.e.\ both horizontal or both vertical), $[v,y]$ and $[y,w]$ have the same orientation (i.e.\ both horizontal or both vertical), different from the orientation of $[x,y]$, and if the resulting state $(P',L')$ is a valid state of Rectangulation. We illustrate the A-rotate and A-flip operations in Figure~\ref{fig:fliprot}.

\begin{figure}[h]
\centering
\includegraphics[width=0.3\textwidth]{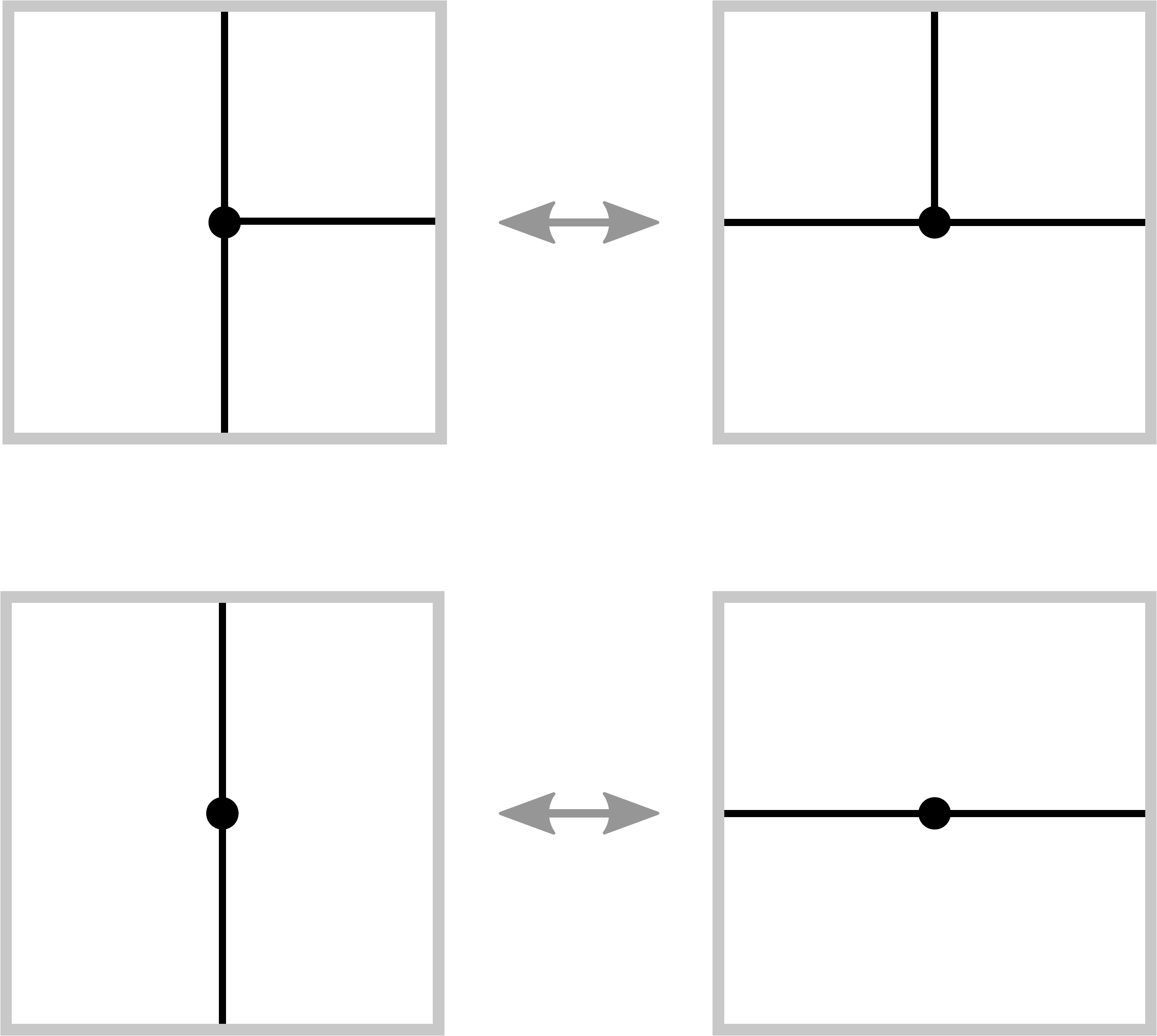}
\caption{(above) A-rotate operations. (below) A-flip operations. \label{fig:fliprot}}
\end{figure}

We make the simple observation that both an A-rotate and an A-flip can be simulated with one, respectively two flip operations as defined in~\textsection\,\ref{sec_r_intro}.

Let $\Rc_1$ and $\Rc_2$ be two valid states of Rectangulation, reachable from the initial state given by a permutation point set $X$. We denote by $d(\Rc_1, \Rc_2)$ the shortest number of A-rotate and A-flip operations that transform $\Rc_1$ to $\Rc_2$. The Flip Diameter problem studied by Ackerman et al.\ asks for the quantity $\diam(X) = \max_{\Rc_1, \Rc_2}{\left\{d(\Rc_1, \Rc_2)\right\}}$, where the maximum is over all valid states of Rectangulation reachable from the initial state determined by $X$. Let $\OPTR(X) = \min_\Ac \left\{ \Ac(X) \right\}$, i.e.\ the smallest cost of any algorithm for the Rectangulation problem with input $X$, as defined in~\textsection\,\ref{sec_r_intro}. We make the following easy observation.

\begin{theorem} \label{old_flip}
 For an arbitrary permutation $X$, we have
$$ \OPTR(X) \leq 2 \cdot \diam(X). $$
\end{theorem}

\begin{proof}
 Let $d'(\Rc_1, \Rc_2)$ be the shortest number of flip operations (according to the definitions in~\textsection\,\ref{sec_r_intro}) that transform $\Rc_1$ to $\Rc_2$. By the observation that two flip operations can simulate an A-rotate or an A-flip, we have that $d'(\Rc_1, \Rc_2) \leq 2 \cdot d(\Rc_1, \Rc_2)$, for all $\Rc_1, \Rc_2$.
 
Let $\diam'(X) = \max_{\Rc_1, \Rc_2}{\left\{d'(\Rc_1, \Rc_2)\right\}}$, where the maximum is over all valid states of Rectangulation reachable from the initial state determined by $X$. By the previous observation, we have $\diam'(X) \leq 2 \cdot \diam(X)$. Furthermore, 
$\diam'(X) \geq \OPTR(X)$, since $\OPTR(X)$ refers to the shortest number of flips between two particular valid states of Rectangulation, the initial state, and an end state. The claim follows. \qedd
\end{proof}

We give a new interpretation of a result of Ackerman et al. They prove the following.

\begin{theorem}[{\protect\cite[\textsection\,3]{Ackerman}}] \label{reverse}
There exists a permutation $X$ of size $n$ such that \ $\diam(X) = \Omega(n \log n)$.
\end{theorem}

The proof of Ackerman et al.\ uses the bitwise reversal permutation $R_n$ (see~\cite{Ackerman} for the definition) and argues about certain geometric constraints that hold for any possible sequence of A-rotate and A-flip operations on rectangulations constrained by this permutation.

We give a very simple alternative proof: It was shown by Wilber in 1989~\cite{Wilber} that for every BST algorithm $\Ac$ it holds that $\Ac(R_n) = \Omega(n \log {n})$. Using the equivalences between the BST problem and Satisfied Superset (Theorem~\ref{lem1}),  respectively, between Satisfied Superset and Rectangulation (Theorem~\ref{thm1}), it follows that $\OPTR(R_n) = \Omega(n \log {n})$. The application of Theorem~\ref{old_flip} finishes the proof.

Ackerman et al.\ raise the open question of computing the average of $\diam(X)$ over all permutation point sets $X$ of size $n$. A simple argument shows that this value is $\Omega(n \log n)$.

\begin{theorem} \label{random}
For a random permutation $X$ of size $n$, we have $\,\mathbb{E}_X\left[\diam(X)\right] = \Omega(n \log n)$.
\end{theorem}

\begin{proof}
It is known~\cite{Wilber, FOCS15} that $\mathbb{E}_X\left[\Ac(X)\right] = \Omega(n \log {n})$ for any BST algorithm $\Ac$. For any algorithm $\Ac'$ for Rectangulation, there is a BST algorithm $\Ac$ such that $\Ac(X) = O(\Ac'(X))$ for all $X$ (Theorem~\ref{lem1} and Theorem~\ref{thm1}). Thus, $\mathbb{E}_X\left[\Ac'(X)\right] = \Omega(n \log {n})$ for every Rectangulation algorithm $\Ac'$. Since $\Ac'(X) \leq 2\cdot \diam(X)$ for some algorithm $\Ac'$ for Rectangulation (Theorem~\ref{old_flip}), the claim follows.
\end{proof}

It would be of interest to describe natural classes of inputs $X$ of size $n$, for which the value of $\diam(X)$ is small, i.e.\ linear in $n$. For the BST problem there has been extensive research on query sequences that can be served with linear total cost~\cite{ST85, in_pursuit, FOCS15, landscape}. For any such sequence $X$ we obtain (via Theorem~\ref{lem1} and Theorem~\ref{thm2}) that $d'(V,H) = O(n)$, where $V$ is the Rectangulation initial state determined by $X$, and $H$ is any Rectangulation valid end state reachable from $V$. 

The claim that easy (linear cost) permutations for the BST problem are also easy (linear cost) for the Flip Diameter problem follows immediately, if the following two conjectures hold. 

\begin{conjecture}
For any two Rectangulation states $\Rc_1$ and $\Rc_2$, we have $d(\Rc_1,\Rc_2) = O\left(d'(\Rc_1,\Rc_2)\right)$.
\end{conjecture}

\begin{conjecture}
For any two Rectangulation states $\Rc_1$ and $\Rc_2$, we have $d'(\Rc_1,\Rc_2) = O\left(d'(H,V)\right)$.
\end{conjecture}

The first conjecture claims that A-rotate and A-flip operations are essentially equivalent with our flip operation, and the second conjecture claims that the distance between the all-vertical and any all-horizontal state is asymptotically the longest of any distances. For instance, Ackerman et al.\ state the open question of whether $\diam(X)$ is linear, if $X$ is a \emph{separable permutation}. Using our recent result that separable permutations are linear-cost for the BST problem~\cite{FOCS15}, if the above conjectures hold, then we get an affirmative answer to the question of Ackerman et al. 

Based on Theorems~\ref{lem1}, \ref{thm1}, and \ref{thm2}, we know that our flip operation between rectangulations captures any possible BST algorithm. It would be interesting to give a characterization of the class of BST algorithms that are captured by the A-flip and A-rotate operations.

In Ackerman et al.~\cite[\textsection\,2]{Ackerman} it is shown that $\diam(X) = O(n \log{n})$, for all permutations $X$ of size $n$. The proof is constructive (i.e.\ an algorithm with worst-case $O(n \log n)$ operations). The proposed algorithm and its analysis are quite sophisticated, for instance, the proof relies on the Four color theorem. However, if we interpret this algorithm in the special case of transforming the all-vertical rectangulation to the all-horizontal rectangulation (i.e.\ for our Rectangulation problem), the output of the algorithm is rather simple: It corresponds to a static balanced binary search tree (whose cost for serving $X$ is clearly $O(n \log{n})$).

We know that A-flip and A-rotate operations can capture non-trivial BST algorithms (i.e.\ other than static trees), since Ackerman et al.\ show that the flip diameter of a diagonal point set is $O(n)$. In the language of binary search trees, this means that the sequence $S = (1,2,\dots,n)$ is accessed in time $O(n)$. An easy counting argument shows that such a bound cannot be achieved by a static BST (since in any BST of size $n$, a constant fraction of the nodes are at depth $\Omega(\log n)$). It is instructive to interpret the algorithm of Ackerman et al.~\cite[\textsection\,4]{Ackerman} given for this particular input in terms of BST rotations (the algorithm corresponds to a straightforward offline BST algorithm tailored for serving the access sequence $S$). In the BST world, such a bound is known to be achieved by several general-purpose algorithms, including Splay trees~\cite{tarjan_sequential} and Greedy~\cite{Fox11}.

\section{Consequences for Small Manhattan Network} \label{sec_manh}
In \textsection\,\ref{sec:ssp} we define the manhattan path between two points $x,y$. Further, we state that given an input permutation point set $X$, a set $Y \supseteq X$ is a solution for the Satisfied Superset problem, iff for all $a,b \in Y$, there is a manhattan path between $a$ and $b$ with respect to $Y$ (\Cref{prop: sat and manhattan}).

An obvious relaxation of the Satisfied Superset problem is to require manhattan paths only between pairs of points from the input point set $X$. This is the Small Manhattan Network problem, which is of independent interest.

More precisely, an algorithm $\Ac$ for Small Manhattan Network outputs, given a permutation point set $X$, a point set $Y \supseteq X$, such that for all $a,b \in X$, there is a manhattan path between $a$ and $b$ with respect to $Y$. The cost of $\Ac$, denoted $\Ac(X)$ is the size of the set $Y$. In the context of the BST problem, Small Manhattan Network was defined by Harmon~\cite{Harmon} as a lower bound for the BST optimum. In a geometric setting, the problem was studied by Gudmundsson, Klein, Knauer, and Smid (GKKS)~\cite{manhattan}.

Let $\OPTM(X) = \min_\Ac \left\{ \Ac(X) \right\}$, i.e.\ the optimum Small Manhattan Network solution for input $X$. Similarly, we define $\OPTS(X)$ the optimum Satisfied Superset solution (which is a constant factor away from the optimum BST solution by Theorem~\ref{lem1}, and from $\OPTR(X)$ by Theorems~\ref{thm1} and \ref{thm2}).

From the above definition of Small Manhattan Network as a less restricted Satisfied Superset problem, the following result is immediate.

\begin{theorem}[\cite{Harmon}] \label{thm4}
For an arbitrary permutation $X$ we have $$\OPTM(X) = O\left(\OPTS(X)\right).$$
\end{theorem}

In \ref{sgr} we describe the Signed Satisfied Superset problem. Let us denote by $\OPTsi(X)$ the Signed Satisfied Superset optimum for $X$, i.e.\ the maximum between the $\boxslash$- and $\boxbslash$-Satisfied Superset optima: $\OPTsi(X) = \max\left\{\OPTSp(X),\ \OPTSm(X)\right\}$. Similarly, if $\OPTRp(X)$ and $\OPTRm(X)$ denote the optima for \elbr- and \elbl-Rectangulation, let us denote by $\OPTri(X)$ the Signed Rectangulation optimum for $X$, i.e.\ $\OPTri(X) = \max\left\{\OPTRp(X),\ \OPTRm(X)\right\}$. \\

We next describe the quantity $\MIR(X)$, the independent rectangle bound, a known lower bound on $\OPTS(X)$ defined by Harmon, DHIKP, and in a similar form by Derryberry, Sleator, and Wang~\cite{Harmon, DHIKP09, DSW}.

We use the definition of the independent rectangle bound $\MIR(X)$ from DHIKP~\cite{DHIKP09}. A pair of points $a,b \in X$ are called an \emph{unsatisfied rectangle} in $X$, if $a$ and $b$ are not on the same horizontal or vertical line, and the rectangle with corners $a$ and $b$ contains no point from $X \setminus \{a,b\}$. Two unsatisfied rectangles $(a,b)$ and $(c,d)$ are independent (in $X$), if none of the four corners of the rectangle with corners $a$ and $b$ is in the interior of the rectangle with corners $c$ and $d$ (and the other way around). Observe that a corner of one rectangle may be on the boundary (or even on the corner) of the other rectangle. The independent rectangle bound is defined as $\MIR(X) = \frac{1}{2}|\I(X)| + |X|$, where $\I(X)$ is the largest set of independent rectangles in $X$. It is known~\cite{Harmon, DHIKP09} that $\MIR(X) = O\left(\OPTS(X)\right)$, for all $X$.\\

We refer as SignedGreedy$(X)$ to the size of the SignedGreedy output for $X$, which is the union of the Greedy$_{\boxslash}$ and Greedy$_{\boxbslash}$ outputs (described in \textsection\,\ref{sgr}). 
 %Let SignedGreedy$(X)$ denote the cost of SignedGreedy for input $X$. 
The following results are known, except for the results involving the new quantity $\OPTri(X)$, which follow from the statements in \textsection\,\ref{sgr}.

\begin{theorem}[\cite{Harmon, DHIKP09}] \label{thm5}
For an arbitrary permutation $X$ we have (up to constant factors):

$$\mathrm{SignedGreedy}(X) ~=~ \OPTsi(X) ~=~ \OPTri(X) ~=~ \MIR(X) ~=~ \OPTM(X).$$

\end{theorem}

For completeness, we give an alternative proof of the statement $\OPTM(X) = \Omega\big(\MIR(X)\big)$. (The other direction  $\OPTM(X) = O\big(\MIR(X)\big)$ follows from the fact that the SignedGreedy output for $X$ is on the one hand, a feasible Manhattan Network solution, and on the other hand, a constant-approximation for $\MIR(X)$.) The proof is inspired by a proof of a similar flavor given by DHIKP~\cite{DHIKP09} for a different statement, and is somewhat simpler than the proof given by Harmon~\cite{Harmon}. 
\begin{proof}
Let $\I_{\boxslash} (X)$ denote the largest set of independent rectangles in $X$ such that the two corner points from $X$ defining each rectangle are such that one is above and to the right of the other. Similarly, let $\I_{\boxbslash} (X)$ denote the largest set of independent rectangles in $X$ such that the two corner points from $X$ defining each rectangle are such that one is above and to the left of the other.

Similarly, define $\OPTML (X)$ the size of the smallest point set $Y \supseteq X$ such that for all pairs of points $a,b \in X$ such that $a$ is above and to the right of $b$, there is a manhattan path between $a$ and $b$ in $Y$. Let $\OPTMR (X)$ the size of the smallest point set $Y \supseteq X$ such that for all pairs of points $a,b \in X$ such that $a$ is above and to the left of $b$, there is a manhattan path between $a$ and $b$ in $Y$. The proof relies on the following lemma.

\begin{lemma} \label{lem:help}
For all permutations $X$, we have $$\OPTML(X) ~\geq~ |X| + |\I_{\boxslash}(X)|,$$   $$\OPTMR(X) ~\geq~ |X| + |\I_{\boxbslash}(X)|.$$
\end{lemma}

\ \\ \noindent As we have $\OPTM(X) ~\geq~ \max\left\{ \OPTML(X), \OPTMR(X) \right\}$, using Lemma~\ref{lem:help}, we obtain: $$\OPTM(X) ~\geq~ |X| + \frac{1}{2} |\I_{\boxslash}(X)| + \frac{1}{2} |\I_{\boxbslash}(X)| ~\geq~ |X| + \frac{1}{2} |\I(X)| ~=~ \MIR(X).$$

It remains to prove Lemma~\ref{lem:help}. We prove the first statement only, as the other statement is entirely symmetric.

Let $\Rc$ be a maximally wide rectangle in $\I_{\boxslash}(X)$, and let $v$ be a vertical line segment with endpoints on the opposite horizontal sides of $\Rc$, such that none of the other rectangles in $\I_{\boxslash}(X)$ intersect $v$. (Such a $v$ exists by the maximality of $\Rc$, and the independence-property of $\I_{\boxslash}(X)$.) %due to the maximality of $\Rc$ and ). 
To simplify the argument, take $v$ such that the $x$- coordinate of $v$ is fractional. Let $a$ and $b$ be the corners of $\Rc$, such that $a$ is above and to the right of $b$. Consider a manhattan path $P_{ab} = (a = x_1, \dots, x_k = b)$, where $x_i \in Y$, and $Y$ is the solution achieving $\OPTM(X)$. Let $p$ and $q$ be the unique neighboring points in $P_{ab}$ such that $p$ is to the left of $v$, and $q$ is to the right of $v$. (Observe that $p$ and $q$ are on the same horizontal line, and there is no point of $Y$ in the interior of $[p,q]$.) Charge the cost of the rectangle $(a,b)$ to the pair $(p,q)$. Remove $\Rc$ from $\I_{\boxslash}(X)$, and continue the process. Observe that the pair $(p,q)$ can not be charged again in the future (since no other rectangle intersects $v$). Furthermore, the number of pairs to which the rectangles can be charged is at most $\OPTML - |X|$. The claim follows. \qedd
\end{proof}

From Theorems~\ref{thm4} and \ref{thm5} it follows that the optimum solution of Small Manhattan Network is a lower bound for every BST solution, and that this lower bound is constant-approximable in polynomial time. Furthermore, it is shown by Harmon and DHIKP that the quantity $\OPTM(X)$ is asymptotically at least as large as two well-known lower bounds given by Wilber~\cite{Wilber} for the BST problem.\\

Equipped with these observations, we revisit the Small Manhattan Network problem studied by GKKS~\cite{manhattan} and reinterpret some of their results. GKKS show the following result.

\begin{theorem}[{\protect\cite[Thm.~1]{manhattan}}]
For any point set $X$ of size $n$ we have \ $\OPTM(X) = O(n \log n)$.
\end{theorem}

The solution given by GKKS is constructive (i.e.\ an algorithm that constructs a manhattan network with $O(n \log n)$ points). We can sketch it as follows: Split $X$ with a vertical line $v$ into two equal subsets, and add the projection of all points in $X$ to $v$ to the solution. Repeat the process recursively on the subsets of $X$ on the two sides of $v$. It is straightforward to verify both that the resulting point set is a valid Manhattan Network solution, and that its size is $O(n \log n)$.

We observe that an alternative way to prove $\OPTM(X) = O(n \log n)$ is simply to note that $\OPTS(X) = O(n \log n)$, and apply Theorem~\ref{thm4}. The upper bound on $\OPTS(X)$ follows from the observation that a BST access sequence can be served with logarithmic cost per access. In fact, it is not hard to see that the algorithm given by GKKS corresponds to the execution trace of a static balanced BST that serves access sequence $X$. (The vertical line $v$ corresponds to the root of the tree, that is touched by every access, and the same holds at every recursive level.) 

GKKS further show the following result.

\begin{theorem}[{\protect\cite[Thm.~4]{manhattan}}]
For some point set $X$ of size $n$ we have \ $\OPTM(X) = \Omega(n \log n)$.
\end{theorem}

The instance used to show this is (essentially) the bitwise reversal sequence $R_n$ mentioned in \textsection\,\ref{sec_rect}. Again, the result can be shown in an alternative way, observing that $\MIR(R_n) = \Omega(n \log n)$, a fact known to Wilber~\cite{Wilber} in 1989, and using the correspondence of Theorem~\ref{thm5}.

The correspondence between Small Manhattan Network and BST yields further results for the Small Manhattan Network problem. In particular, for the BST problem we have several fine-grained bounds on the cost of the optimum solution, such as dynamic finger, working set, or the traversal bound~\cite{ST85}. Results of this type give immediate upper bounds on the complexity of the Small Manhattan Network solution for inputs with particular structure. (Although some of these structures may seem unusual in a geometric setting.) 

Similarly to the result for Flip Diameter, we obtain the following.

\begin{theorem} \label{random_mn}
For a random point set $X$ of size $n$, the complexity of the Small Manhattan Network optimum is $\Theta(n \log n)$.
\end{theorem}

Since only the relative ordering of the points matters for Small Manhattan Network (and not the distances between points), by random point set we mean a point set in general position whose relative ordering corresponds to a random permutation. Again, the proof only needs the result of Wilber~\cite{Wilber} and its extension to permutations~\cite{FOCS15}.

Recently we showed~\cite{FOCS15} that for all permutations $X$ of size $n$ that forbid an arbitrary constant-size permutation pattern, it holds that $\MIR(X) = O(n)$. This yields the following observation.

\begin{theorem} \label{random_mn}
Every planar point set that avoids a fixed permutation pattern admits a manhattan network of linear complexity.
\end{theorem}

\section{Consequences for the BST problem}\label{sec3}

The ``flip'' and ``tree-relax'' models of the BST problem (described in \textsection\,\ref{sec2}) give new interpretations of several well-studied concepts in the BST world. We list some preliminary observations and questions in this direction.

\paragraph{Upper bounds.}

For an arbitrary pair $(u,v)$ of points, where $u = (u.x, u.y)$, and $v = (v.x, v.y)$, let us define the \emph{height} of $(u,v)$ as $h(u,v) = |u.y - v.y|$, and the \emph{width} of $(u,v)$ as $w(u,v) = |u.x - v.x|$.  For an arbitrary monotone tree $T$, let $h(T)$ be the sum of heights of all edges in $T$, and let $w(T)$ be the sum of widths of all edges in $T$. 

Consider a BST access sequence $X$ of size $n$ (a permutation), and the corresponding treap $T$ on $X$, as well as the path $P$ on $X$ (both are defined in \textsection\,\ref{sec_tr}). Consider an edge-flip operation $(a \rightarrow b)$ in some tree $T'$ that adds the edge $(a,b)$ and removes the edge $(r,b)$, where $r$ is the parent of $b$ in $T'$. Let the resulting tree be $T''$. We make the following two observations: 
\begin{align*}
h(T') - h(T'') = h(r,b) - h(a,b) = h(r,a) &~~\geq~~ 1, \mbox{~~~~~~and} \\
w(T') - w(T'') = w(r,b) - w(a,b) = -w(r,a) &~~\leq~~ -1.
\end{align*}

In words, the total height strictly decreases, and the total weight strictly increases in every edge-relax operation. We also observe that $h(P) = n-1$ (in the end, every edge is of height 1). It follows that the quantities $H = h(T) - h(P) = h(T) - n + 1$, and $W = w(P) - w(T)$ are upper bounds on the cost of \emph{every} algorithm for the Tree Relaxation problem. 
Given $X$, both $W$ and $T$ can be easily computed. The bounds are however, not very strict, as both $W$ and $T$ can be as large as $\Theta(n^2)$. Nevertheless, for certain highly structured inputs, such as for permutations close to the sequential access $(1,\dots,n)$, the quantities are asymptotically tight bounds for the BST problem.

We can strengthen both bounds, by summing the \emph{logarithms} of the heights, respectively weights. More precisely, we define for an arbitrary monotone tree $T$ the quantities 
$$H'(T) = \displaystyle\sum_{(u,v) \in T} {\log{\big( h(u,v) \big)}},$$
$$W'(T) = \displaystyle\sum_{(u,v) \in T} {\log{\big( w(u,v) \big)}}.$$

If $T$ is the initial treap on $X$, and $P$ is the path on $X$, then $W'(P)$ is the classical \emph{dynamic finger} bound, and $H'(T)$ is (essentially) the classical \emph{working set} bound~\cite{ST85}. (Working set is typically defined in the literature with respect to the last occurrence of the \emph{same element} in an access sequence. However, for permutation access sequences it is natural to consider the occurrence of the nearest successor or predecessor among the already seen elements, which is exactly what the quantity $H'(T)$ captures.) The bounds $W'(P)$ and $H'(T)$ no longer hold for \emph{every} algorithm, but $W'(P)$ is known to be asymptotically matched by certain algorithms, e.g.\ Greedy~\cite{Fox11, LI16} and Splay tree~\cite{ColeMSS00,Cole00}. 

A different upper bound on the cost of every Tree Relaxation algorithm can be computed by summing for all vertices in the monotone tree, the distance to the root. (By distance we mean the number of edges on the path to the root.) Again, it can be seen that this quantity strictly increases with ever edge-flip operation, reaching in the end the value $n(n-1)/2$.

\paragraph{New heuristics.}

The above quantities suggest natural greedy heuristics for Tree Relaxation. For instance, in every step we may perform the edge-flip that decreases the total edge height the most, or that increases the total edge width the most, or that increases the total distance-from-the-root the most. We leave for further research the question of how efficient (and how natural) the corresponding BST algorithms are.

In the Rectangulation problem we flip from the all-vertical to the all-horizontal state. Natural measures of quality for any intermediate state include the total length of remaining vertical segments, the total length of horizontal segments, or the difference between the two quantities. It would seem natural to perform flips that greedily optimize any of these quantities. It remains open whether the resulting BST algorithms are efficient.

\paragraph{Interpretations of Greedy.}

to be added soon

\newpage 
\bibliographystyle{plain}
\bibliography{ref}

\newpage 
\appendix 
\clearpage

\end{document}

%% file: signed.tex
In this section, we show that a relaxed version of the Satisfied Superset
problem, called Signed Satisfied Superset~\cite{DHIKP09} is equivalent to the relaxed version of the Rectangulation problem, called Signed Rectangulation.

\paragraph{Signed Satisfied Superset.}

We recall the definition of Signed Satisfied Superset from DHIKP~\cite{DHIKP09}. A point set $Y\subseteq[n]\times[n]$ is \emph{$\boxslash$-satisfied}
if for any two points $a,b\in Y$ where $a_{x}<b_{x}$ and $a_{y}<b_{y}$,
the rectangle with corners $a$ and $b$ contains some point in $Y\setminus\{a,b\}$,
possibly on the boundary of the rectangle. We similarly say that $Y$ is 
\emph{$\boxbslash$-satisfied} if the condition holds for all $a,b\in Y$
where $a_{x}<b_{x}$ and $a_{y}>b_{y}$. Note that $Y$ is
satisfied iff $Y$ is $\boxslash$-satisfied and $\boxbslash$-satisfied.

Given a permutation point set $X$ of size $n$, an algorithm $\Ac$
for the $\boxslash$- ($\boxbslash$-) Satisfied Superset problem outputs a point set
$Y$ with $X\subseteq Y\subseteq[n]\times[n]$, such that $Y$ is
$\boxslash$-satisfied ($\boxbslash$-satisfied). The cost of $\Ac$ is the
size of the set $Y$, denoted $\Ac(X)$. We generally call
the $\boxslash$- and $\boxbslash$-Satisfied Superset problems \emph{Signed} Satisfied Superset. (The symbols $\boxslash$ and $\boxbslash$ can be read as ``plus'' and ``minus''.)

Let $\OPTSp(X)$ and $\OPTSm(X)$ be the optimum cost for the $\boxslash$-, resp.\ $\boxbslash$-Satisfied Superset problem for $X$. DHIKP~\cite{DHIKP09} show
that the natural variants of Greedy called Greedy$_{\boxslash}$ and
Greedy$_{\boxbslash}$ return the optimum solution, i.e.\ $\mbox{Greedy}_{\boxslash}(X)=\OPTSp(X)$
and $\mbox{Greedy}_{\boxbslash}(X)=\OPTSm(X)$. Let SignedGreedy~\cite{Harmon, DHIKP09} be the algorithm
that returns the union of the Greedy$_{\boxslash}$ and Greedy$_\boxbslash$ solutions.

While it is clear that $\max\left\{\OPTSp(X),\ \OPTSm(X)\right\}\le\OPTS(X)$ for
every $X$, it is a long-standing open problem whether $\max\left\{\OPTSp(X),\ \OPTSm(X)\right\}=\Omega(\OPTS(X))$.
If true, this would imply that SignedGreedy is a polynomial-time 
algorithm for constant-approximating $\OPTS(X)$ (and in light of Theorem~\ref{lem1}, the BST optimum as well).

\paragraph{Signed Rectangulation.}

Let $[p,q]$ be a vertical segment, and $[q,r]$ be a horizontal segment.
We say that $[p,q]$ and $[q,r]$ form a \emph{\elbl-elbow} iff (i)
$p$ is above $q$ and $r$ is on the right of $q$, or (ii) $p$
is below $q$ and $r$ is on the left of $q$. Symmetrically, we say
that $[p,q]$ and $[q,r]$ form a \emph{\elbr-elbow} iff (i) $p$
is above $q$ and $r$ is on the left of $q$, or (ii) $p$ is below
$q$ and $r$ is on the right of $q$.

A state $(P,L)$ of the Rectangulation problem is \emph{\elbl-elbow-free}, respectively \emph{\elbr-elbow-free} iff each non-margin point in $P$ is contained
in at least two segments of $L$, and if it is contained in exactly
two segments, then they must not form a \elbl-elbow, resp.\ \elbr-elbow. 

We define the \elbl-Rectangulation problem the same way as Rectangulation, 
except that we only require that each state $(P,L)$ of the \elbl-Rectangulation
problem is \emph{\elbr-elbow-free} instead of elbow-free (i.e.\ the \elbl elbows are allowed). We 
similarly define the \elbr-Rectangulation problem. We generally call the \elbl-
and \elbr-Rectangulation problems \emph{Signed} Rectangulation. (The symbols \elbl and \elbr can be read as ``plus'' and ``minus''.)

Given any set $S$ of allowed elbows, we can similarly define the $S$-Rectangulation problem in an obvious way. For example, \elba-Rectangulation problem or \elbtr-Rectangulation problem.

\begin{theorem} \label{thm: siRec gives siSat}
	Any algorithm $\Ac$ for the \textrm{\elbr}- or \textrm{\elbl}-Rectangulation problem
	can be transformed (in polynomial time) into an algorithm $\Ac'$
	for the $\boxslash$-, respectively $\boxbslash$-Satisfied Superset problem, such that for all
	inputs $X$, we have $\Ac'(X)=O(\Ac(X))$.\end{theorem}
\begin{proof}
	We only show the case of \elbl-Rectangulation. The other case is symmetric.
	The proof goes in the same way as in \Cref{thm1}. 
	Initially, let $Y=X$. We construct an algorithm $\Ac'$ from $\Ac$ by adding to $Y$ every non-margin endpoint created while flipping from the all-vertical to the all-horizontal state in Rectangulation. The cost of $\Ac$ is equal to the number of flips. Since each flip adds at most two points to $Y$, the claim on the cost of $\Ac'$ is immediate. 
	
	We claim that $Y$ is $\boxbslash$-satisfied. Suppose otherwise, that there are two points $a,b\in Y$ where $a$ is above and to the left of $b$. 
	Let $\left<a,a'\right>$ be the last flip in the execution of $\Ac$ such that $a'$ is on the same horizontal line as $a$ and to the right of $a$. Let $\left<b',b\right>$ be the last flip such that $b'$ is on the same horizontal line as $b$ and to the left of $b$.
	(There have to be such flips, otherwise $\Ac$ would not produce a valid end state.) Since the rectangle with corners $a,b$ is empty, $b'$ must be to the left of $a$, and $a'$ must be to the right of $b$. 
	
	Suppose that the flip $\left<b',b\right>$ occurs earlier than the flip $\left<a,a'\right>$ (the other case is symmetric), and consider the state before the flip $\left<a,a'\right>$. In that state there must be a vertical segment with top endpoint at $a$, otherwise $a$ would be contained in 
	\elbr-elbow. (This is the only difference from the proof of \Cref{thm1}.)
	Let $a^*$ be the bottom endpoint of the vertical segment with top endpoint $a$. The point $a^*$ must be strictly below $b$, for otherwise the rectangle with corners $a,b$ would contain it. This means that $[a,a^*]$ intersects $[b',b]$, contradicting that we are in a valid state. 
	We conclude that $Y$ is a $\boxbslash$-satisfied superset of $X$. \qedd 
\end{proof}
\begin{theorem}\label{thm: siSat gives siRec}
	Any algorithm $\Ac'$ for the $\boxslash$-Satisfied Superset
	problem can be transformed (in polynomial time) into an algorithm
	$\Ac$ for the \elbdr-Rectangulation problem, such that
	for all inputs $X$, we have $\Ac(X)=O(\Ac'(X))$.\end{theorem}
\begin{proof}
	We only show the case for $\boxslash$-Satisfied Superset problem. The other case is symmetric.
	The proof goes in the same way as in \Cref{thm2}.
	Let $Y$ be a $\boxslash$-satisfied set constructed by $\Ac'$.
	We construct an algorithm $\Ac$ that maintains the state $(P,L)$ with the following operations in a greedy manner:
	(1) if some valid flip $\left<a,b\right>$ is possible where $a,b \in Y \cup M$, then execute it, and
	(2) if some vertical segment $[a,b]\in L$ containing no point from $Y$ (except possibly its endpoints) can be removed, then remove it.
	Here, the valid flip is defined according to the \elbdr-Rectangulation problem.
	We claim that $\Ac$ reaches an end state. The cost of $\Ac$ follows with the same argument as in \Cref{thm2}.
	Suppose for contradiction that $\Ac$ gets stuck at an intermediate state $(P,L)$.
	
	As in \Cref{thm2}, consider two points $q, q' \in Y \cup M$ on the same horizontal line, $q$ to the left of $q'$, such that $[q,q']$ is not in $L$, and the segment $[q,q']$ contains no point of $Y$ in its interior. If there is no such pair of points, then we are done, since all horizontal lines are complete, and all remaining vertical segments can be removed. Among such pairs, consider the one where $q$ is the rightmost, in case of a tie, choose the one where $q'$ is the leftmost. 
	Left-extensibility and right-extensibility are defined as in the proof of \Cref{thm2}.
	
	Observe that throughout the execution of $\Ac$, for any state $(P,L)$, every point in $Y$ is contained in some segment of $L$. Since $\left<q,q'\right>$ is not a valid flip, $[q,q']$ must intersect some vertical line $[z,z'] \in L$ (assume w.l.o.g.\ that $z$ is strictly above, and $z'$ is strictly below $[q,q']$). Observe that $[z,z']$ cannot contain a point of $Y$ in its interior. If it would contain such a point $z^*$, then $z^*$ would be the left endpoint of some segment missing from $L$, contradicting the choice of $q$. Thus, since removing $[z,z']$ is not a valid step according to \elbdr-Rectangulation problem, it must be that either $z$ is left-extensible, $z$ is right-extensible or $z'$ is right-extensible.
	If $z$ or $z'$ were right-extensible, that would contradict the choice of $q$. Therefore, $z$ is left-extensible.
	
	Since $Y$ is $\boxslash$-satisfied, by the statement analogous to \Cref{prop: sat and manhattan} there has to be a point $w \in Y \setminus \{z,q\}$ either on the horizontal segment $[(q.x,z.y),z]$, or on the vertical segment $[z,(z.x,q.y)]$. 
	Since $[z,z']$ cannot contain a point of $Y$ in its interior, it must be the case that $w$ is on $[(q.x,z.y),z]$, and choose $w$ to be closest to $z$. But then the segment $[w,z]$ is missing from $L$, contradicting the choice of $q,q'$ because $w.x \geq q.x$. \qedd
	 
\end{proof}
Theorems analogous to \Cref{thm: siSat gives siRec} for \elbul,\elbur,\elbdl-Rectangulation can be shown similarly.
By \Cref{thm: siRec gives siSat,thm: siSat gives siRec}, we have that (i) $\boxslash$-Satisfied Superset problem, \elbr,\elbdr,\elbul-Rectangulation problems are equivalent, and (ii) $\boxbslash$-Satisfied Superset problem, \elbl,\elbdl,\elbur-Rectangulation problems are equivalent.

Let us also consider the case of allowing two types of elbows that are neighbors in the clockwise ordering of the four possible elbows, i.e.\ the \elba, \elbb, \elbc, \elbd-Rectangulation problems. Together with \elbl and \elbr, these are all possible cases with two types of allowed elbows. We argue that \elba-Rectangulation is trivial: For every input of size $n$ there is a flip sequence of length $O(n)$. Due to the symmetries of the problem, the same holds for the \elbb, \elbc, \elbd ~cases, and consequently, also for Rectangulation with three or four types of allowed elbows.

The algorithm for obtaining a linear sequence of flips for \elba-Rectangulation is as follows. First execute an inital phase as in the proof of Theorem~\ref{tree_rect}, then complete the horizontal rectangulation line by line, from top to bottom. At step $k$, assume that the horizontal lines $k,\dots,n$ are completed. Remove every vertical segment whose top endpoint is at height $k$ (observe that this can only create \elba-elbows). Then, complete the horizontal line at height $k-1$, by flipping horizontal segments at height $k-1$ to the maximum extent possible (this can not create crossings, since we removed vertical segments in the previous step).